\documentclass[article,onecolumn]{IEEEtran}
\usepackage{amsthm}
\usepackage{amsmath}
\usepackage{amssymb}

\usepackage{enumitem}
\usepackage{dsfont}
\usepackage[noadjust]{cite}

\newtheorem{theorem}{Theorem}

\newtheorem{lemma}{Lemma}

\theoremstyle{definition}

\newtheorem*{example*}{Example}
\newtheorem*{remark*}{Remark}

\newtheorem{remark}{Remark}

\def\pe{p_e}

\newcommand{\RealF}{\mathds{R}}

\newcommand{\Expt}{\mathds{E}}

\newcommand{\eps}{\varepsilon}

\newcommand{\wt}[1]
{\widetilde{#1}}

\def\ind{\mathds{1}}
\def\wt{\widetilde}
\newcommand{\dfn}{\triangleq}

\def\mod1{\;\mathrm{mod}\;1}

\begin{document}

\title{A Simple Proof for the Optimality of Randomized Posterior Matching}
\author{Ofer~Shayevitz and Meir~Feder\thanks{The authors are with the Department of EE--Systems, Tel Aviv University, Tel Aviv, Israel \{ofersha@eng.tau.ac.il, meir@eng.tau.ac.il\}. The work of O. Shayevitz was supported by the Israel Science Foundation, grant no. 1367/14.}}

\date{}

\maketitle

\thispagestyle{empty}

\begin{abstract}
Posterior matching (PM) is a sequential horizon-free feedback communication scheme introduced by the authors, who also provided a rather involved optimality proof showing it achieves capacity for a large class of memoryless channels. Naghshvar \textit{et al} considered a non-sequential variation of PM with a fixed number of messages and a random decision-time, and gave a simpler proof establishing its optimality via a novel Shannon-Jensen divergence argument. Another simpler optimality proof was given by Li and El Gamal, who considered a fixed-rate fixed block-length variation of PM with an additional randomization. Both these works also provided error exponent bounds. However, their simpler achievability proofs apply only to discrete memoryless channels, and are restricted to a non-sequential setup with a fixed number of messages. In this paper, we provide a short and transparent proof for the optimality of the fully sequential horizon-free PM scheme over general memoryless channels. Borrowing the key randomization idea of Li and El Gamal, our proof is based on analyzing the random walk behavior of the shrinking posterior intervals induced by a reversed iterated function system (RIFS) decoder. 
\end{abstract}

\section{Introduction}\label{sec:introduction}
Posterior Matching (PM) is a simple and general feedback communication scheme introduced by the authors, who also showed it achieves capacity for a large class of memoryless channels, including discrete alphabets, continuous alphabets, and mixtures thereof \cite{posterior_matching_IT,posterior_matching,posterior_matching_isit08}. One appealing feature of the PM scheme is that it is horizon-free and sequential, in the sense that the transmitter may send an infinite sequence of bits, and the receiver can decide to stop at every instant $n$; the receiver is then able decode roughly $nC$ bits from the prefix of this sequence with vanishing error probability, where $C$ is the capacity of the channel. Alternatively, the receiver is also able to decode the bits on the fly as soon as they become reliable enough. As argued in \cite{posterior_matching_IT}, PM can easily be converted to the more traditional settings where the number of messages and/or the horizon are fixed.  

While heuristic arguments for the optimality of PM are simple and appealing (see \cite{posterior_matching_IT}, and going back to the special case of the Horstein scheme \cite{Horstein-report,horstein}), the original optimality proof in \cite{posterior_matching_IT} is quite involved and nontransparent. Coleman \cite{coleman09} studied the PM scheme from a novel stochastic control and Lyopanov exponent perspective, and provided a conceptually cleaner approach for its analysis. Naghshvar \textit{et al} \cite{naghshvar2013extrinsic} considered a non-sequential variation of PM restricted to discrete memoryless channels (DMCs), where the number of messages is fixed but the decision time (horizon) is random. Introducing a novel Shannon-Jensen divergence, they provided a simpler proof showing that their scheme achieves the capacity of any DMC. Li and El Gamal \cite{LiElGamal:2015} considered the same setting but with a fixed horizon. They described a randomized variation of PM and provided a simpler proof showing it achieves the capacity of any DMC. A key ingredient in their scheme was a random shift applied to the message point after each PM iteration, which circumvented some of the analysis obstacles. Both \cite{naghshvar2013extrinsic} and \cite{LiElGamal:2015} also provide error exponent results. 

In this paper, we adopt the random shift idea of Li and El Gamal, and consider a randomized version of the fully sequential horizon-free PM scheme. We provide a short and transparent optimality proof, showing that this scheme achieves the capacity for a very large class of memoryless channels, including all DMC and also many continuous alphabet and mixed alphabet channels. Our proof is based on analyzing the random-walk behavior of a reversed iterated function system (RIFS) decoder introduced in \cite{posterior_matching_IT}. Unlike the deterministic PM scheme in \cite{posterior_matching_IT}, the combination of RIFS decoding and the random shift operation facilitates a much cleaner analysis and avoids the problem of fixed points that was a major obstacle in the original proof.

\section{Preliminaries}
\subsection{Definitions and Basic Lemmas}
Recall that a real-valued stochastic process $T_n$ is called a submartingale if $\Expt(T_{n+1}\mid T^n) \geq T_n$ for any $n$. The following result is well known. 
\begin{lemma}[Martingale Convergence Theorem \cite{doob1953stochastic}]\label{lem:martingale_conv}
Let $T_n$ be a submartingale. If $\sup_n\Expt|T_n| < \infty$ then $T_n$ convergence a.s. to some r.v. $T$ and $\Expt|T| < \infty$. 
\end{lemma}

Let $g:[0,1]\mapsto\mathbb{R}$ be a Lebesgue measurable function. With some abuse of notations, we naturally extend $g$ to operate on subsets of its domain in an element-wise fashion, namely $g(A) \dfn \cup_{x\in A}\{g(x)\}$ for any set $A\subseteq [0,1]$. We write $|A|$ for the Lebesgue measure of the set $A$, whenever the former exists. Define the $\lambda$-smoothed derivative of $g$ to be
\begin{align*}
  D_\lambda [g(x)] \dfn \frac{1}{\lambda}\left|g\left(\left[x -\tfrac{\lambda}{2}, x+\tfrac{\lambda}{2}\right] \mod1\right)\right| , 
\end{align*} 
where $t \mod1 \dfn t-\lfloor t\rfloor$ is the modulo $1$ operation.\footnote{One may equivalently identify $[0,1)$ with the circle $\mathbb{R}\slash \mathbb{Z}$, in lieu of the modulo notation. The cyclic definition of the smoothed derivative takes care of what happens near the edges of the unit interval, and is essential for our purposes later due to the random shift. The definition (and associated results in this section) work with minor adaptations for any other interval domains (with the proper modulo) or when the domain is $\mathbb{R}$ (without the modulo).} Let 
\begin{align*} 
D[g(x)] \dfn \limsup_{\lambda\to 0}D_\lambda [g(x)] .
\end{align*}
The following lemma is easily verified. 
\begin{lemma}\label{lem:derivative}
  If $g(x)$ is differentiable at $x_0\in(0,1)$ with a derivative $g'(x_0)$, then $D[g(x_0)] = |g'(x_0)|$. Furthermore, if $g$ is absolutely continuous on $[0,1]$, then  
\begin{align*}
  D_\lambda [g(x)] = \Expt\left|g'\left((x+ Q_\lambda) \mod1 \right)\right| ,
\end{align*} 
where $Q_\lambda\sim \textrm{Unif}\left(\left[-\tfrac{\lambda}{2},\tfrac{\lambda}{2}\right]\right)$. 
\end{lemma}

Now, further define 
\begin{align*}
  \overline{D}[g(x)] \dfn \sup_{\lambda\in(0,1)}D_\lambda [g(x)] .
\end{align*}
When $g$ is absolutely continuous and monotonic (which will be our case of interest), then $\overline{D}[g(x)]$ is the maximal stretching of any symmetric interval (modulo $1$) around $x$ by $g$. The following lemma is a consequence of the Hardy-Littlewood maximal inequality \cite{rudin1987real}, and states that $\overline{D}[g(x)]$ is unlikely to be too large, provided that $g$ is well behaved. The proof is relegated to the appendix. 
\begin{lemma}\label{lem:hardy_littlewood}
Let $g:[0,1]\mapsto\mathbb{R}$ be absolutely continuous on $[0,1]$, and $X\sim \textrm{Unif}([0,1])$. Then for any $a>0$,
  \begin{align*}
    \Pr\left(\overline{D}[g(X)] > a\right) \leq 9a^{-1}\Expt\left|g'(X)\right| .
  \end{align*}
\end{lemma}
\begin{remark}
  Note that if $g$ is Lipschitz (which corresponds in the sequel to the case of discrete alphabet channels), then a stronger asymptotic statement trivially holds: $\Pr\left(\overline{D}[g(X)] > a\right) = 0$ for all $a$ large enough. 
\end{remark}

Let $(X,Y)\sim P_{XY}$ be jointly distributed real-valued random variables. Let $F_X$ be the c.d.f. of $X$, and $F_X^{-1}$ be its functional inverse, generally defined by 
\begin{align*}
  F_X^{-1}(v) \dfn \inf\{x:F_X(x)>v\} . 
\end{align*}
It is easy to verify (see e.g. \cite{posterior_matching_IT}) that we can always define an auxiliary r.v. $\Theta\sim\textrm{Unif}([0,1])$ such that $X=F^{-1}_X(\Theta)$. This induces a joint distribution $P_{\Theta XY}$. Let $F_{\Theta\mid Y}(\theta\mid y)$ denote the conditional c.d.f. of $\Theta$ given $Y$, also known as the \textit{PM kernel} \cite{posterior_matching_IT}. We will also be interested in the \textit{inverse PM kernel} $F^{-1}_{\Theta\mid Y}(v\mid y)$, which is the functional inverse of the PM kernel w.r.t. $\theta$ \cite{posterior_matching_IT}. 

In the remainder of the paper, we restrict our attention to the following family $\mathfrak{F}$ of all distributions $P_{XY}$ admitting the following two properties: 
\begin{enumerate}[label=(\textbf{P\arabic*})]
\item \label{P1}$F_{\Theta\mid Y}(\theta\mid y)$ (resp. $F_{\Theta\mid Y}^{-1}(v\mid y)$) is absolutely continuous and strictly monotone in $\theta\in[0,1]$ (resp. $v\in[0,1]$) for $P_Y$-a.a. $y$. 
\item \label{P2} There exists some $\delta>0$ such that 
  \begin{align*}
    \lim_{\lambda\to 0}\Expt|\log D_\lambda [F^{-1}_{\Theta\mid Y}(V\mid Y)]|^{2+\delta}  < \infty,
  \end{align*}
where  $Y\sim P_Y$ and $V\sim \textrm{Unif}([0,1])$ are independent, and the $\lambda$-smoothed derivative is taken w.r.t. $v$. 
\end{enumerate}
\begin{remark}\label{rem:family}
  The family $\mathfrak{F}$ is quite rich and includes all discrete distributions, as well as many continuous and mixed alphabet distributions. See Remark \ref{rem:family2} following Theorem \ref{thrm:main}. 
\end{remark}

The following claims are readily verified. 
\begin{lemma}\label{lem:PM_kernel}
  Suppose $P_{XY}$ satisfies property \ref{P1}. Then 
  \begin{enumerate}[label=(\roman*)]
  \item $\frac{\partial}{\partial v} F_{\Theta\mid Y}^{-1}(v\mid y) = 1/f_{\Theta\mid Y}(F_{\Theta\mid Y}^{-1}(v\mid y)\mid y)$ for $P_Y$-a.a. $y$. \label{item:PM_der}
  \item $I(X;Y) = I(\Theta;Y) < \infty$. \label{item:Ixy=Ithetay}
  \end{enumerate}
\end{lemma}

Finally, we say that a r.v. $X$ is \textit{stochastically smaller} than another r.v. $Y$, if $\Pr(Y\leq a) \leq \Pr(X\leq a)$ for any $a$. More generally, we say that $X$ is stochastically smaller than $Y$ given some event $A$, if $\Pr(Y\leq a \mid  A) \leq \Pr(X\leq a)$ for any $a$.

\subsection{Setup} 
We are concerned with the following feedback communication setup. A transmitter is in possession of a \textit{message point} $\Theta_0\sim\textrm{Unif}([0,1])$, its binary expansion representing an infinite i.i.d. uniform bit sequence to be reliably communicated to a receiver over a memoryless channel $P_{Y|X}$. The input and output of the channel at time $n$ are denoted $X_n$ and $Y_n$ respectively. We assume there is a noiseless instantaneous feedback link from the receiver back to the transmitter, so that at time $n$ the transmitter is in possession of $Y^{n-1}$. The memoryless channel model means that $Y_n$ is independent of $(X^{n-1},Y^{n-1},\Theta_0)$ given $X_n$, and that $Y_n\mid X_n=x_n \sim P_{Y\mid X}(\cdot\mid x_n)$. Furthermore, we assume the transmitter and the receiver share some common randomness; specifically, we assume they can jointly draw an i.i.d. sequence $\{V_n\sim\textrm{Unif}\left(\left[0,1\right]\right)\}_{n=1}^\infty$, where $V_n$ is statistically independent of $(\Theta_0,X^n,Y^n,V^{n-1})$. 

A (sequential, horizon-free) \textit{transmission scheme} is an infinite sequence of mappings that determine the next channel input $X_{n+1}$ as a function of $(\Theta_0,Y^n,V^n)$. A \textit{decoding rule} is a corresponding sequence of functions that map  $(Y^n,V^n)$ to an interval (modulo $1$) $J_n$, in which the receiver believes the message point lies. The \textit{error probability} attained by a scheme and a decoding rule at time $n$ is $p_e=\Pr(\Theta_0\not\in J_n)$, and the associated \textit{instantaneous rate} is $R_n = -\frac{1}{n}\log |J_n|$. The relation to decoding actual bits is simple: Identifying the said interval of size $2^{-nR_n}$ essentially guarantees that the $nR_n$ most significant bits of $\Theta_0$ can be decoded with error probability $p_e$, up to technical edge issues that can be easily resolved (see \cite{posterior_matching_IT}). A transmission scheme is said to attain a rate $R$, if for any target error probability $p_e>0$ there is a suitable decoding rule such that $\Pr(R_n \geq  R)\to 1$ as $n\to\infty$. In the following two subsections we describe a simple and optimal construction of a transmission scheme and decoding rule, namely the randomized PM scheme with  RIFS decoding.  

\subsection{Randomized Posterior Matching}

Let $P_{Y\mid X}$ be a memoryless channel law, and set some input distribution $P_X$ (say, capacity achieving under some input constraint). Consider the following recursively defined transmission scheme: 
\begin{align}\label{eq:rpm}
  &\Theta_1 = \Theta_0 \nonumber\\
  &X_n = F_X^{-1}(\Theta_n) \nonumber \\ 
  &\Theta_{n+1} = \left(F_{\Theta\mid Y}\left(\Theta_n \mid Y_n\right)+ V_n\right) \mod1
\end{align}
The scheme in~\eqref{eq:rpm} will be referred to as the \textit{randomized PM scheme}. Note that for $V_n=0$ this coincides with the classical PM scheme \cite{posterior_matching_IT}. The randomization idea is key to our simplified analysis, and is due to Li and El Gamal \cite{LiElGamal:2015} who analyzed a non-sequential fixed-rate fixed-block-length version of this scheme in a DMC setting.  

We recall a few known properties of PM that are also inherited by its randomized sibling, with minor modifications accounting for common randomness. The proofs follow easily from the associated claims in \cite{posterior_matching_IT}, e.g. by thinking of $(Y_n,V_n)$ as the channel output, and are omitted. 

\begin{lemma}\label{lem:PM_props}
  The randomized PM scheme satisfied the following:
  \begin{enumerate}[label=(\roman*)]
  \item $\Theta_n\sim \textrm{Unif}\left([0,1]\right)$, $X_n\sim P_X$, and $Y_n\sim P_Y$.
  \item $\Theta_n$ (and hence $X_n$) is statistically independent of $(Y^{n-1},V^{n-1})$.
  \item $\{Y_n\}_{n=1}^\infty$ and $\{V_n\}_{n=1}^\infty$ are mutually independent i.i.d. sequences. \label{item:YViid}
  \item $I(\Theta_0; Y_n\mid Y^{n-1},V^n) = I(X;Y)$.
  \item $I(\Theta_0; Y^n \mid  V^n) = nI(X;Y)$.
  \end{enumerate}
\end{lemma}

\subsection{Reversed Iterated Function System (RIFS) Decoding}
In this subsection we describe a decoding rule for the randomized PM, that maps $Y^n$ into an interval that is guaranteed to contain the message point $\Theta_0$ up to a prescribed error probability (see \cite{posterior_matching_IT} for more details). Let $F^{-1}_{\Theta\mid Y}(v\mid y)$ be the inverse PM kernel, i.e., 
\begin{align*}
  F_{\Theta\mid Y}^{-1}(v\mid y) \dfn \inf\{\theta:F_{\Theta\mid Y}(\theta\mid y)>v\}.
\end{align*}

Set some target error probability $\pe>0$, and let $J_0\subset(0,1)$ be an interval of size $|J_0| = 1-\pe$. The RIFS decoder outputs the interval $J_n$ defined recursively by 
\begin{align}\label{eq:RIFS_Jn}
  J_{k+1} &= F^{-1}_{\Theta\mid Y}(\left(J_k - V_{n-k}\right) \mod1 \mid Y_{n-k})
\end{align}
for $k=0,\ldots,n-1$. Recall that we effectively identify $[0,1)$ with the circle $\mathbb{R}\slash \mathbb{Z}$, hence we allow wrap-around intervals, i.e., the interval $(a,b)$ for $a>b$ is the union $(a,1)\cup[0,b)$. 
\begin{lemma}[\cite{posterior_matching_IT}]\label{lem:RIFS}
  The probability of error incurred by the above RIFS decoder is $\Pr(\Theta_0\not\in J_n) = p_e$. 
\end{lemma}
\begin{proof}
\begin{align}
  \Pr(\Theta_0\in J_n) &= \Pr(\Theta_1\in J_n) \nonumber\\
  &= \Expt\Pr(\Theta_1\in J_n\mid Y_1,V_1) \nonumber\\
  &= \Expt\Pr(\Theta_2\in J_{n-1}\mid Y_1,V_1) \label{eq:RIFS_inv}\\
  &= \Pr(\Theta_2\in J_{n-1}) \label{eq:RIFS_ind}\\ 
  &= \cdots \label{eq:RIFS_iterate}\\
  &= \Pr(\Theta_n\in J_0) \nonumber\\ 
  &= 1-\pe .\label{eq:RIFS_bydef}
\end{align}
\eqref{eq:RIFS_inv} follows since~\eqref{eq:RIFS_Jn} is invertible given $Y_{n-k},V_{n-k}$, by virtue of property~\ref{P1}. ~\eqref{eq:RIFS_ind} follows since by Lemma~\ref{lem:PM_props} $\Theta_{k+1}$ is independent of $(Y_k,V_k)$. In~\eqref{eq:RIFS_iterate} we iterate the same arguments, and~\eqref{eq:RIFS_bydef} holds by definition. 
\end{proof}

Define the sequence of \textit{contraction terms}: 
\begin{align*}
  L_k \dfn \log\left(\frac{|J_{k-1}|}{|J_k|}\right) ,
\end{align*}
and set $L_0\dfn -\log(1-p_e)$. Define further 
\begin{align*}
  R_n \dfn \frac{1}{n}\sum_{k=0}^nL_k . 
\end{align*}

From the discussion above it is clear that the RIFS decoder outputs an interval of (random) size $2^{-nR_n}$ in which $\Theta_0$ is guaranteed to lie with probability $1-p_e$. Therefore, $R_n$ is the (random) instantaneous rate of randomized PM under RIFS decoding with error probability $p_e$. In what follows, we will be interested in guarantees on $R_n$. As we shall see, in many cases $R_n$ becomes arbitrarily close (for any target $p_e$) to the optimal value $I(X;Y)$ with high probability as $n$ grows large. Thus, randomized PM can achieve any rate up to channel capacity. 

\section{Main Result}

We state our main result, showing that under very mild regularity conditions the randomized PM scheme with RIFS decoding achieves any rate below the mutual information. 

\begin{theorem}\label{thrm:main}
Let $(X,Y)\sim P_{XY}\in\mathfrak{F}$ and assume that $0<I(X;Y)<\infty$. Then for any target error probability $\pe$ and any $\eps>0$, the decoding rate achieved by the associated randomized PM scheme with RIFS decoding satisfies 
\begin{align*}
\lim_{n\to\infty}\Pr(R_n > I(X;Y)-\eps)  = 1
\end{align*}
\end{theorem}

\begin{remark}\label{rem:family2}
  The conditions in the theorem are very general, and specifically hold in the following cases:
  \begin{itemize}
  \item For any discrete memoryless channel with any input distribution such that $I(X;Y)>0$. In this case \cite{posterior_matching_IT} the PM kernel is a quasi-linear function in $\theta$ for any fixed $y$, with slopes corresponding to the conditional distributions of $x$ given $y$. 
  \item When the conditional p.d.f. $f_{X|Y}(x|y)$ exists, is bounded, and has bounded support, for any $y$. 
  \item For any additive noise channel $Y=X+Z$ where $Z$ is independent of $X$, both $Z$ and $Y$ have bounded p.d.fs, and either:
    \begin{itemize}
    \item $f_Z(z), f_Y(y)$ have bounded supports; or, 
    \item $f_Z(z) \geq 2^{-O(|z|^{k_1})}, f_Y(y) \geq 2^{-O(|y|^{k_2})}$ and $\Expt |Z|^{3k_1}, \Expt |Y|^{3k_2} < \infty$ for some $k_1,k_2>0$. This includes in particular the additive Gaussian channel with a Gaussian input, where the scheme essentially reduces to the well known Schalwijk-Kailath Scheme \cite{Schalkwijk-Kailath,Schalkwijk2}. Note that this subfamily also includes mixed alphabet channels, e.g. binary input and additive Gaussian noise, etc. 
  \end{itemize}
\end{itemize}
\end{remark}

\begin{remark}
  The original PM optimality result (no randomization) requires the posterior matching kernel to be free of any fixed points \cite{posterior_matching_IT}. It was further shown in \cite{posterior_matching_allerton} that the existence of such fixed points is possible, and that in such a case no positive rate can be attained, unless a suitable input transformation is applied. We note that the randomized PM does not suffer from this issue; the fixed point problem is ``washed away'' by the random shifting operation. 
\end{remark}

\section{Proof of  Main Result}
\subsection{Proof Sketch}
Before we proceed to formally prove Theorem \ref{thrm:main}, we give a heuristic argument that captures the essence of the proof. Let $S_n\dfn nR_n  = \sum_{k=0}^nL_k$ be the sum of contraction terms at time $n$. First, note that if we fix the horizon $n$, the process $\{S_k\}_{k=1}^n$ is a Markov chain in the time index $k$. Alas, the stochastic process $S_n$ is \textit{not} a Markov chain in the horizon parameter $n$, since the RIFS process evolves backward in time (see \cite{posterior_matching_IT} for more details). However, since we are only interested in the asymptotic (marginal) behavior of $S_n$ as the horizon $n$ grows unbounded, then instead of fixing the horizon $n$ and analyzing the process $S_k$, we can assume the horizon is infinite and think of $S_n$ as a Markov chain for any $n\in\mathbb{N}$ (with some abuse of notations, where we replaced $S_k$ with $S_n$). The associated processes $L_n$ and $J_n$ will be indexed by $n$ as well. In other words, we are effectively thinking of the decoding process going forward in time, instead of backward. 

How does the process $S_n$ evolve? At time $n$, imagine we are in possession of some random interval $J_n$ of size $|J_n| = 2^{-S_n}$, corresponding to the interval the RIFS holds after $n$ backward iterations. The position of $J_n$ is uniformly distributed over the unit interval modulo $1$, due to the random shift operation. We independently draw a r.v. $Y_n\sim P_Y$ (recalling that the output sequence is i.i.d), and apply the inverse PM kernel to obtain the next interval $J_{n+1} = F^{-1}_{\Theta\mid Y}(J_n \mid Y_n)$, which is then randomly shifted modulo $1$. This procedure yields the update 
\begin{align*}
  S_{n+1} = S_n + L_n, \quad \textrm{where} \quad  L_n = \log\left(\frac{|J_n|}{|J_{n+1}|}\right).
\end{align*}
The process $S_n$ is thus a Markovian random walk on $\RealF^+$, starting from $S_0 = -\log(1-p_e)$, with the contraction terms $L_n$ as its increments.

Now, assume that $S_n$ is already very large, i.e. that the associated interval size $|J_n|$ is very small. What is the increment $L_n$ in this case? Clearly, $J_n$ will shrink (or stretch) by a (random) factor that is roughly the derivative of $F^{-1}_{\Theta\mid Y}(v\mid y)$ w.r.t. $v$, evaluated for $y=Y_n$ and at $v$ that is (say) the random midpoint of $J_n$, which is $\sim\textrm{Unif}\left([0,1]\right)$ and independent of $Y_n$. By Lemma \ref{lem:PM_kernel} claim \ref{item:PM_der}, this derivative is equal to $1/f_{\Theta\mid Y}(F^{-1}_{\Theta\mid Y}(v\mid y))$. The contraction term is hence roughly $\log f_{\Theta\mid Y}(F^{-1}_{\Theta\mid Y}(V_n\mid Y_n)\mid Y_n)$. Defining $\tilde{\Theta} = F^{-1}_{\Theta\mid Y}(V_n\mid Y_n)$, it is readily verified that $(\tilde{\Theta},Y_n)\sim P_{\Theta Y}$ as induced by $P_{XY}$ and $X=F_X^{-1}(\Theta)$ (see Lemma \ref{lem:Llim0}). Thus, we conclude that when $S_n$ is large, the contraction term $L_n$ has distribution close to that of the r.v. $\log f_{\Theta\mid Y}(\Theta\mid Y)$, and hence $\Expt L_n \approx I(\Theta;Y) = I(X;Y)$. Thus, as long as $S_n$ does not become too small, it grows like the sum of roughly i.i.d. random variables with expectation $I(X;Y)$, which is why we expect $S_n$ to be close to $n I(X;Y)$. 

Of course, the devil is in the details. The main technical challenge is to bound the behavior of the chain for small $S_n$, in which case the contraction terms behave quite differently; in contrast to the case of a large $S_n$ where the distribution of the contraction terms is essentially independent of the actual value of $S_n$, here this distribution strongly depends on the exact position of the random walk. More specifically, instead of being the logarithm of the derivative of the inverse PM kernel, the contraction terms in the ``small'' regime correspond to the logarithm of the $\lambda$-smoothed derivative of the inverse PM kernel, with a smoothing factor of $\lambda=2^{-S_n}$. In the next subsection, we deal with these difficulties: First, we show that $S_n$ spends overall little time in the ``small'' regime (note that it can go back and forth between ``large'' and ``small''). Then, we couple the process $S_n$ with a simpler process $S_n'$ that has only two modes of i.i.d. behavior, corresponding to whether $S_n$ is ``small'' or ``large''. We show that the contribution of the ``small'' mode of $S_n'$ is negligible, and that consequently $S_n'$ is close to $nI(X;Y)$ with high probability. The proof is then completed by observing that $S_n'$  is stochastically smaller than $S_n$.

\subsection{Detailed Proof}
In this subsection we prove Theorem \ref{thrm:main}. We use the definition of $S_n$ as a Markovian random walk on $\RealF^+$, with the time arrow going forward instead of backward, as described in the previous subsection. Define the random variable 
\begin{align}\label{eq:Llambda}
L^{(\lambda)}\dfn -\log D_\lambda [F^{-1}_{\Theta\mid Y}(V\mid Y)] ,
\end{align}
where  $Y\sim P_Y$ and $V\sim \textrm{Unif}([0,1])$ are independent. Clearly, the distribution of $L^{(\lambda)}$ is the same as the distribution of the contraction factor $L_n$ given that $S_{n-1} = -\log\lambda$. 

We begin by proving two lemmas characterizing the behavior of $L^{(\lambda)}$. 
\begin{lemma}\label{lem:Llim0}
Let $\wt{\Theta} \dfn  F_{\Theta\mid Y}^{-1}(V\mid Y)$. Then $(\wt{\Theta},Y)\sim P_{\Theta Y}$ and 
\begin{align*}
  \lim_{\lambda\to 0}L^{(\lambda)} = \log \frac{f_{\Theta\mid Y}(\wt{\Theta}\mid Y)}{f_\Theta(\wt{\Theta})} \quad a.s. 
\end{align*}
\end{lemma}
\begin{proof}
By assumption \ref{P1}, Lemma \ref{lem:derivative}, and Lemma \ref{lem:PM_kernel} claim \ref{item:PM_der}, we have that given $V=v$ and $Y=y$ 
\begin{align*}
  \lim_{\lambda\to 0} -\log D_\lambda [F^{-1}_{\Theta\mid Y}(v\mid y)]  &= -\log \frac{\partial}{\partial v}\left(F^{-1}_{\Theta\mid Y}(v\mid y)\right) \\
  &= \log f_{\Theta\mid Y}(F_{\Theta\mid Y}^{-1}(v\mid y)\mid y) \\ 
  & =  \log \frac{f_{\Theta\mid Y}(F_{\Theta\mid Y}^{-1}(v\mid y)\mid y)}{f_\Theta(F_{\Theta\mid Y}^{-1}(v\mid y))}
\end{align*}
for $P_{VY}$-a.a. $(v,y)$, where the last step follows trivially since $f_\Theta(\theta) = 1$ for any $\theta\in(0,1)$. It follows that $L^{(\lambda)}$ converges a.s. to the random variable $\log f_{\Theta\mid Y}(\wt{\Theta}\mid Y)$, where $\wt{\Theta}$ is defined in the Lemma. Now 
\begin{align}
  \Pr(\wt{\Theta}\leq \theta \mid  Y=y) &= \Pr(F_{\Theta\mid Y}^{-1}(V\mid Y)\leq \theta \mid  Y=y) \nonumber \\
  & = \Pr(V\leq F_{\Theta\mid Y}(\theta\mid y) \mid  Y=y) \label{eq:inv1PM} \\
  & = F_{\Theta\mid Y}(\theta\mid y) ,\label{eq:YindV}
\end{align}
where~\eqref{eq:inv1PM} holds due to the strict monotonicity of the PM kernel under assumption \ref{P1}, and~\eqref{eq:YindV} follows since $Y$ and $V$ are independent. Hence,  $(\wt{\Theta},Y)\sim P_{\Theta Y}$ according to the joint distribution induced by $(P_X,P_{Y\mid X})$. This completes the proof. 
\end{proof}

\begin{lemma}\label{lem:Llambda}
$\Expt L^{(\lambda)}$ satisfies the following properties:
  \begin{enumerate}[label=(\roman*)]
  \item $\Expt L^{(\lambda)}$ is continuous in $\lambda$ over $[0,1]$.
  \item $\lim_{\lambda\to 1}\Expt L^{(\lambda)} = 0$. 
  \item $\lim_{\lambda\to 0}\Expt L^{(\lambda)} = I(X;Y)$. \label{item:EL=I}
  \item If $I(X;Y)>0$ then $0 <  \Expt L^{(\lambda)} < I(X;Y)$ for any  $\lambda\in(0,1)$. \label{item:ELpositive}
  \end{enumerate}
\end{lemma}
\begin{proof}
The first claim follows easily from assumption \ref{P1}, by the continuity of the inverse PM kernel. The second claim holds since $F^{-1}(\cdot\mid y)$ maps the unit interval to itself for any $y$. Let us prove the third claim. By property \ref{P2} of the family $\mathfrak{F}$, there must exists some $\lambda_0 > 0$ such that $\{L^{(\lambda)}\}_{\lambda\in(0,\lambda_0)}$ is bounded in $\mathcal{L}^p$ for $p = 2+\delta > 1$. Hence $\{L^{(\lambda)}\}_{\lambda\in(0,\lambda_0)}$ are uniformly integrable. By Lemma \ref{lem:Llim0}, $L^{(\lambda)}$ also converges a.s. to a finite limit. Thus, by Vitali's convergence theorem \cite{rudin1987real}, we can change the order of limit and expectation, i.e., 
  \begin{align*}
    \lim_{\lambda\to 0}\Expt L^{(\lambda)} &= \Expt\lim_{\lambda\to 0} L^{(\lambda)} \\ 
    & = \Expt \log\frac{f_{\Theta\mid Y}(\Theta\mid Y)}{f_\Theta(\Theta)} \\
    & = I(\Theta;Y) \\ 
    & = I(X;Y), 
  \end{align*}
where we have used Lemma \ref{lem:PM_kernel} claim \ref{item:Ixy=Ithetay} in the last step. 

For the fourth claim, note that we can write 
\begin{align*}
L^{(\lambda)} = -\log \Expt_Q\left(1/f_{\Theta\mid Y}(F_{\Theta\mid Y}^{-1}(\left(V + Q\right)\mod1\mid Y)\mid Y)  \right) , 
\end{align*}
where $Q\sim \textrm{Unif}([-\tfrac{\lambda}{2},\tfrac{\lambda}{2}])$ is independent of $V,Y$. We therefore have that 
\begin{align}
\nonumber \Expt L^{(\lambda)} &= \Expt_{V,Y}L^{(\lambda)} \\
\nonumber&= -\Expt_{V,Y}\log \Expt_Q\left(1/f_{\Theta\mid Y}(F_{\Theta\mid Y}^{-1}(\left(V + Q\right)\mod1\mid Y)\mid Y)\right)   \\
&< \Expt_{V,Y}\Expt_Q\log f_{\Theta\mid Y}(F_{\Theta\mid Y}^{-1}(\left(V + Q\right)\mod1\mid Y)\mid Y)   \label{eq:Jensenupper} \\
\nonumber& = \Expt_{V',Y}\log f_{\Theta\mid Y}(F_{\Theta\mid Y}^{-1}(V'\mid Y)\mid Y)   \\
& = \Expt_{\Theta Y}\log f_{\Theta\mid Y}(\Theta\mid Y)   \label{eq:seePrevLem}\\
\nonumber& = I(\Theta;Y) \\ 
& = I(X;Y) , \label{eq:Ibound}
\end{align}
where $V'=(V+Q)\mod1$ is uniform over the unit interval. We have used Jensen's inequality in~\eqref{eq:Jensenupper}, which is strict since $\lambda>0$ and $I(\Theta;Y)>0$. \eqref{eq:seePrevLem} follows from Lemma \ref{lem:Llim0},  and~\eqref{eq:Ibound} follows again from Lemma \ref{lem:PM_kernel} claim \ref{item:Ixy=Ithetay}. Similarly,
\begin{align}
\nonumber \Expt L^{(\lambda)} &= -\Expt_{V,Y}\log \Expt_Q\left(1/f_{\Theta\mid Y}(F_{\Theta\mid Y}^{-1}(\left(V + Q\right)\mod1\mid Y)\mid Y)\right)   \\
& > -\log\Expt_{V,Y}\Expt_Q \left(1/f_{\Theta\mid Y}(F_{\Theta\mid Y}^{-1}(\left(V + Q\right)\mod1\mid Y)\mid Y)\right)   \label{eq:Jensenlower} \\
\nonumber& = -\log \Expt_{V',Y} \left(1/f_{\Theta\mid Y}(F_{\Theta\mid Y}^{-1}(V'\mid Y)\mid Y)\right)   \\
\nonumber& = -\log \Expt_{\Theta Y} \left(1/f_{\Theta\mid Y}(\Theta\mid Y)\right)   \\
\nonumber& = -\log \Expt_Y \Expt_{\Theta\mid Y}\left(1/f_{\Theta\mid Y}(\Theta\mid Y)\right)   \\
\nonumber& = -\log \Expt_Y 1   \\
& = 0 .
\end{align}
\end{proof}

Using the properties of $L^{(\lambda)}$ established above, we would like to show that $S_n$ spends little time close to the origin. To that end, we first prove a the following lemma. 
\begin{lemma}\label{lem:limsup}
  $S_n$ is a submatrigale on $\mathbb{R}^+$, and $\Pr(\limsup_{n\to\infty} S_n = \infty) = 1$. 
\end{lemma}
\begin{proof}
The submartingale claim follows immediately from Lemma \ref{lem:Llambda} property \ref{item:ELpositive}. Let us prove the other  claim. Recall that by Lemma \ref{lem:Llambda}, $\Expt L^{(\lambda)}$ is a continuous function of $\lambda$ over $[0,1]$, and $0<\Expt L^{(\lambda)} < I(X;Y)$ for any $\lambda\in(0,1]$, where the upper and lower bounds are approached as $\lambda$ tends to zero and one respectively. It is therefore easy to construct a two-sided monotonically decreasing sequence $\{\lambda_k\}_{k=-\infty}^\infty$ with $\lim_{k\to-\infty}\lambda_k = 1$ and $\lim_{k\to\infty}\lambda_k = 0$ such that 
\begin{align}\label{eq:lambda_interval}
\inf_{\lambda\in[\lambda_{k+1},\lambda_k)}\Expt L^{(\lambda)} > 3\log\frac{\lambda_k}{\lambda_{k+1}}
\end{align}
for any $k$. Hence,  
\begin{align}
\nonumber \delta_k &\dfn \inf_{\lambda\in[\lambda_{k+1},\lambda_k)}\Pr\left(L^{(\lambda)} > 2\log\frac{\lambda_k}{\lambda_{k+1}}\right) \\
& \label{eq:delta_k1} \geq \inf_{\lambda\in[\lambda_{k+1},\lambda_k)}\Pr\left(L^{(\lambda)} > \frac{2}{3}\inf_{\lambda'\in[\lambda_{k+1},\lambda_k)}\Expt L^{(\lambda')}\right) \\ 
& \label{eq:delta_k2} \geq \inf_{\lambda\in[\lambda_{k+1},\lambda_k)}\Pr\left(L^{(\lambda)} > \frac{2}{3}\Expt L^{(\lambda)}\right) \\
& \label{eq:delta_k3}> 0, 
\end{align}
where~\eqref{eq:delta_k1} follows from~\eqref{eq:lambda_interval}, choosing  $\lambda'=\lambda$ establishes~\eqref{eq:delta_k2}, and~\eqref{eq:delta_k3} trivially holds since $\Expt L^{(\lambda)} > 0$ on any closed subinterval of $(0,1)$.   

Let $\{\tau_{j,k}\}_{j=1}^{T_k}$ be the sequence of all time indices $n$ where $S_n\in (-\log\lambda_k,-\log\lambda_{k+1}]$, where $T_k$ is the (possibly infinite) total number of such occurrences.  Let $M_k$ be the maximal time index $n$ for which $S_n > -\log\lambda_{k+1}$, and let $b$ be some fixed positive integer. 
\begin{align*}
  \Pr\left(\limsup_{n\to\infty} S_n \in (-\log\lambda_k,-\log\lambda_{k+1}]\right)  &= \Pr\left(M_k<\infty, T_k = \infty\right) \\
&\leq \Pr\left(M_k<\infty, T_k \geq M_k + b\right) \\ 
& =\sum_{m=0}^\infty\Pr\left(T_k\geq m+b \mid M_k=m\right)\Pr(M_k=m)
\\ &\leq \sum_{m=0}^\infty\Pr\left(L_{\tau_{j,k}}\leq \log\frac{\lambda_k}{\lambda_{k+1}}, m< j\leq m+b \mid M_k=m\right)\Pr(M_k=m) \\
& \leq \sum_{m=0}^\infty(1-\delta_k)^b\Pr(M_k=m)\\
& \leq (1-\delta_k)^b .
\end{align*}
Since $\delta_k>0$, and as the above upper bound holds for any $b$ and $k$, it must be that 
\begin{align*}
\Pr\left(\limsup_{n\to\infty} S_n \in (-\log\lambda_k,-\log\lambda_{k+1}]\right) = 0 .
\end{align*}
The proof is now concluded by noting that $\RealF^+ = \bigcup_k(-\log\lambda_k,-\log\lambda_{k+1}]$. 
\end{proof}

We now further strengthen Lemma~\ref{lem:limsup} and show that $S_n$ in fact diverges a.s., which will specifically show that it spends little time below any threshold $t$. Let $N_{t,n}$ be the number of times $S_k$ falls below $t$ until time $n$, i.e., 
\begin{align*}
  N_{t,n} \dfn \sum_{k=1}^n\ind(S_k < t), 
\end{align*}
and let $N_t\dfn \lim_{n\to\infty}N_{t,n}$ be a random variable on $\mathbb{N}\cup \{\infty\}$.  
\begin{lemma}\label{lem:SDiverges}
  $S_n\to\infty$ almost surely, hence $\Pr(N_{t,n} > m) \leq \Pr(N_t > m)  = \delta(m)$ where $\delta(m)\to 0$ as $m\to\infty$.
\end{lemma}
\begin{proof}
The proof is based on arguments similar to \cite{lamperti1960criteria}. Consider the process $T_n = 1-\frac{1}{1+S_n}$. Below we show that $T_n$ converges a.s., which together with Lemma~\ref{lem:limsup} implies that that $T_n\to 1$ a.s. and hence $S_n\to\infty$ a.s., establishing the lemma. 

First, we show it is sufficent to prove that there exists some $t_0\in(0,1)$ such that $\Expt(T_{n+1}\mid T_n=t) \geq t$ for any $t\geq t_0$. To see that, define the process $T_n' = \max(T_n,t_0)$, and note that by definition it holds that $\Expt(T_{n+1}'\mid T_n'=t) \geq t$ for any $t$, hence $T_n'$ is a submartingale.  Moreover, $\Expt|T_n'|\leq 1$ for all $n$. By Lemma \ref{lem:martingale_conv}, it must therefore be that $T_n'$ convergences a.s. to a limit. Since $\Pr(\limsup_{n\to\infty}T_n'=1) \geq \Pr(\limsup_{n\to\infty}T_n=1) = 1$, this limit must be $1$, i.e., $T_n'\to 1$ a.s . Since $T_n=T_n'$ whenever $T_n'\geq t_0$, it must be that $T_n\to 1$ a.s. as well. 

It remains to show the existence of such a $t_0$. Let us first establish some guarantees on the first and second moments of $L^{(\lambda)}$, conditioned on an event that $L^{(\lambda)} > a$ for some $a$. From Lemma \ref{lem:Llambda} we know that $\Expt L^{(\lambda)}$ approaches $I(X;Y)>0$ continuously as $\lambda\to 0$, hence in particular there is some $c_1>0$ such that $\Expt L^{(\lambda)} > c_1$ for all $\lambda>0$ small enough. Trivially, it also holds that for any $a$ 
\begin{align}\label{eq:L1Lbound}
  \Expt \left(L^{(\lambda)} \mid L^{(\lambda)} > a\right) \geq  \Expt L^{(\lambda)} > c_1 > 0
\end{align}
for any $\lambda>0$ small enough. Moreover, property \ref{P2} of the family $\mathfrak{F}$ implies that $L^{(\lambda)}$ is uniformly bounded in $\mathcal{L}^2$ for all $\lambda>0$ small enough, hence $\Expt |L^{(\lambda)}|^2  < c_2 $ for some some $c_2<\infty$. Trivially then, for any $a$ it also holds that 
\begin{align}\label{eq:L2Lbound}
  \Pr(L^{(\lambda)} > a)\cdot \Expt \left(|L^{(\lambda)}|^2 \mid L^{(\lambda)} > a\right) \leq \Expt \left(|L^{(\lambda)}|^2\right) < c_2 < \infty  
\end{align}
for all $\lambda>0$ small enough. 

Now, define the function $g(s,\ell) \dfn \tfrac{1}{1+s} - \tfrac{1}{1+s+\ell}$. Since the process $S_n$ is nonnegative, we can clearly limit our discussion to $\ell\geq -s$, and hence to $g(s,\ell) \geq -1$. Let us write 
\begin{align*}
  g(s,\ell) &= \frac{\ell}{(1+s)^2 + \ell(1+s)} \\ 
  & = \frac{\ell}{(1+s)^2}  - \frac{\ell^2}{(1+s)^3+\ell(1+s)^2} . 
\end{align*}
Setting any $\alpha\in(0,1)$, it therefore holds that for any $\ell\geq -(1+s)^\alpha$ and $s>2^{\frac{1}{1-\alpha}}-1$, 
\begin{align}\label{eq:gsl_bound}
\nonumber  g(s,\ell) &\geq \frac{\ell}{(1+s)^2}  - \frac{\ell^2}{(1+s)^3-(1+s)^{2+\alpha}} \\
  &\geq  \frac{\ell}{(1+s)^2}  - \frac{\ell^2}{2(1+s)^3} . 
\end{align}
Our analysis will now naturally depend on the event $L_n \geq -(1+s)^\alpha$. Let us first upper bound the probability of the complementary event:
\begin{align}
\nonumber   \Pr(L_n <  -(1+s)^\alpha\mid  S_n=s) &\leq \Pr(|L_n| > (1+s)^\alpha \mid  S_n=s) \\
\nonumber  &= \Pr(|L_n|^{2+\delta} > (1+s)^{\alpha(2+\delta)} \mid  S_n=s) \\ 
\label{eq:LnMarkov}  &\leq \frac{\Expt \left(|L_n|^{2+\delta}\mid S_n=s\right)}{(1+s)^{\alpha(2+\delta)}} \\
\nonumber  & = \frac{\Expt \left(\left|L^{(2^{-s})}\right|^{2+\delta}\right)}{(1+s)^{\alpha(2+\delta)}}\\
\label{eq:LBoundedNorm}& \leq c_3\cdot (1+s)^{-\alpha(2+\delta)} 
\end{align}
for some $c_3>0$ and any $s$ large enough. We used Markov's inequality in \eqref{eq:LnMarkov}, and~\eqref{eq:LBoundedNorm} is again by virtue of property \ref{P2} of the family $\mathfrak{F}$, that implies $L^{(\lambda)}$ is uniformly bounded in $\mathcal{L}^{2+\delta}$ for all $\lambda>0$ small enough. 

Writing $t=1-\tfrac{1}{1+s}$ we have that for any $s$ sufficiently larger than $2^{\frac{1}{1-\alpha}}-1$
\begin{align}
\nonumber   \Expt(T_{n+1}-T_n\mid T_n=t) &= \Expt \left(g(s,L_n) \mid S_n=s\right) \\ 
\nonumber & = \Expt \left(g(s,L^{(2^{-s})})\right) \\
\nonumber   &=\Pr\left(L^{(2^{-s})}< -(1+s)^{\alpha}\right)\cdot\Expt \left(g(s,L^{(2^{-s})}) \mid  L^{(2^{-s})}< -(1+s)^{\alpha}\right) \\ 
\nonumber &\quad + \Pr\left(L^{(2^{-s})}\geq -(1+s)^{\alpha}\right)\cdot\Expt \left(g(s,L^{(2^{-s})}) \mid L^{(2^{-s})}\geq -(1+s)^{\alpha}\right)\\ 
\label{eq:dT1} & \geq  -c_3\cdot (1+s)^{-\alpha(2+\delta)} \\ 
\nonumber & \quad + \Pr\left(L^{(2^{-s})}\geq -(1+s)^{\alpha}\right)\cdot \Expt\left(\frac{L^{(2^{-s})}}{(1+s)^2} \mid L^{(2^{-s})} \geq  -(1+s)^{\alpha}\right)\\
\nonumber & \quad - \Pr\left(L^{(2^{-s})}\geq -(1+s)^{\alpha}\right)\cdot \Expt\left(\frac{\left|L^{(2^{-s})}\right|^2}{2(1+s)^3} \mid L^{(2^{-s})} \geq  -(1+s)^{\alpha}\right) \\ 
\label{eq:ETbound} & \geq  -c_3\cdot (1+s)^{-\alpha(2+\delta)} + \left(1-c_3\cdot (1+s)^{-\alpha(2+\delta)})\right)\cdot \frac{c_1}{(1+s)^2} - \frac{c_2}{2(1+s)^3} . 
\end{align}
\eqref{eq:dT1} follows from~\eqref{eq:gsl_bound},~\eqref{eq:LBoundedNorm}, and since $g(s,\ell) \geq -1$. \eqref{eq:ETbound} follows from~\eqref{eq:L1Lbound},~\eqref{eq:L2Lbound}, and~\eqref{eq:LBoundedNorm}. Examining~\eqref{eq:ETbound} for any $\tfrac{2}{2+\delta}<\alpha < 1$, it is immediately clear that this lower bound on the expected increment is positive for all large enough $s$, and hence for all $t$ sufficiently close to $1$. This concludes the proof.

\end{proof}

After establishing that $S_n\to\infty$ a.s., we would like to further determine how fast this happens. To that end, we will define a coupled process $S_n'$ that will be easier to handle, and will be stochastically smaller than $S_n$. Loosely speaking, $S_n'$ will have two modes of i.i.d. random walk behavior corresponding to whether $S_n$ is above or below the threshold $t$; it will also grow slower than $S_n$ in each of these regimes. 

To do that, we first define two random variables $U,W$ that will be stochastically smaller than $L_n$ given that $S_n$ is above or below the threshold $t$ respectively, and will later determine the increments of the coupled process $S_n'$ in these two regimes. For brevity, we omit the dependence of $U,W$ on $t$. Recall the definition of $L^{(\lambda)}$ in~\eqref{eq:Llambda}. We first define $\wt{U},\wt{W}$ via their c.d.fs as follows: 
\begin{equation*}
\Pr(\wt{U}\leq u) \dfn \sup_{\lambda\in(0,2^{-t}]}\Pr\left(L^{(\lambda)} \leq u\right) , 
\end{equation*}
\begin{equation*}
\Pr(\wt{W}\leq w) \dfn \sup_{\lambda\in(2^{-t},1)}\Pr\left(L^{(\lambda)} \leq w\right) . 
\end{equation*}
Now, setting some large number $\xi>0$, we define $U,W$ as the truncation of $\tilde{U},\tilde{W}$:
\begin{equation*}
U\dfn \min(\wt{U},\xi)\,,\quad W\dfn \min(\wt{W},\xi) .
\end{equation*}
Again, the dependence on $\xi$ will be omitted for notational clarity. The following lemma describes some important properties of $U$ and $W$. The proof is relegated to the appendix.

\begin{lemma}\label{lem:uw}
The following properties hold:
\begin{enumerate}[label=(\roman*)]
\item $U$ is stochastically smaller than $L_n$ given $S_{n-1} = t_0$ for any $t_0\geq t$ \label{item:u}
\item $W$ is stochastically smaller than $L_n$ given $S_{n-1} = t_0$ for any $t_0<t$\label{item:w}
\item $\Expt U \leq I(X;Y)$ for any $t,\xi$. \label{item:ubound} 

\item $\lim_{\xi\to\infty}\lim_{t\to\infty}\Expt U = I(X;Y)$. \label{item:ulim} 
\item $\Expt |W| < \infty$ for any $\xi,t > 0$.  \label{item:wlim}
\end{enumerate}
\end{lemma}

We are now ready to define the coupled process $S_n'$. Let $\{U_n\}$ and $\{W_n\}$ be two i.i.d. sequences with distributions $P_U$ and $P_W$ respectively, such that the processes $\{U_n\},\{W_n\},\{S_n\}$ are mutually independent. Define $S_n'$ to be the random walk process generated by replacing the increments of the process  $S_n$ process with $U$ or $W$ elements, according to whether $S_n$ is above or below the threshold. Precisely: 
\begin{align*}
  S_n' = \sum_{k=1}^{n-N_{t,n}} U_k + \sum_{k=1}^{N_{t,n}} W_k .
\end{align*}
Note that unlike $S_n$, the coupled process $S_n'$ can become negative, since $\Pr(W\leq 0) = 1$. Also, $S_n'$ does not contain the fixed initialization term $L_0=-\log(1-p_e)$. The proof of the following lemma appears in the appendix. 
\begin{lemma}\label{lem:lower_bound_sn}
$S_n'$ is stochastically smaller than $S_n$.
\end{lemma}

Let us now show the probability $S_n'$ falls below $n(I(X;Y)-\eps)$ vanishes with $n$. 
\begin{lemma}\label{lem:S'achievable}
$\lim_{n\to\infty} \Pr(S_n' > n(I(X;Y)-\eps)) = 1$ for any $\eps>0$. 
\end{lemma}
\begin{proof}
We write $I=I(X;Y)$ for short. Set $\xi$ and $t$ large enough so that such that 
\begin{align}\label{eq:UclosetoI}
I - \Expt U \leq \eps/8 , 
\end{align}
which is possible by virtue of Lemma \ref{lem:uw} claims \ref{item:ulim} and \ref{item:ubound}. Then:
\begin{align}
\nonumber \Pr(S_n' < n(I-\eps)) &=\Pr\left(\sum_{k=1}^{n-N_{t,n}} U_k + \sum_{k=1}^{N_{t,n}} W_k < n(I-\eps)\right) \\ 
\nonumber &\leq \Pr(N_{t,n} > m)  + \sum_{r=1}^m\Pr(N_{t,n}=r)\Pr\left(\sum_{k=1}^{n-r} U_k + \sum_{k=1}^r W_k < n(I-\eps) \mid N_{t,n} = r\right) \\ 
\label{eq:Sbound1} &\leq \delta(m) + \sum_{r=1}^m\Pr(N_{t,n} = r)\left[\Pr\left(\sum_{k=1}^{n-r}U_k < nI-\frac{n\eps}{2}  \;\bigvee\; \sum_{k=1}^rW_k < -\frac{n\eps}{2}\right)\right] \\
\label{eq:Sbound2}&\leq \delta(m) + \sum_{r=1}^m\Pr(N_{t,n} = r)\left[\Pr\left(\sum_{k=1}^{n-r}U_k < nI-\frac{n\eps}{2}\right) + \Pr\left(\sum_{k=1}^rW_k < -\frac{n\eps}{2}\right)\right] \\
\label{eq:Sbound3} & \leq \delta(m) + \sum_{r=1}^m\Pr(N_{t,n} = r)\left[\Pr\left(\frac{1}{n-r}\sum_{k=1}^{n-r}U_k < \frac{I - \frac{\eps}{2}}{1-\frac{r}{n}}\right) + \Pr\left(\frac{1}{r}\sum_{k=1}^rW_k < -\frac{n\eps}{2r} \right)\right] .
\end{align}
\eqref{eq:Sbound1} follows from Lemma \ref{lem:SDiverges} and since the sequences $\{U_n\},\{V_n\}$ are mutually independent of $\{S_n\}$, hence of $N_{t,n}$ as well. \eqref{eq:Sbound2} follows from the union bound. Analyzing the first term  inside the parenthesis in~\eqref{eq:Sbound3}, we note that for any $1\leq r\leq m$ and $n>m$ large enough, 
\begin{align}
\nonumber   \Pr\left(\frac{1}{n-r}\sum_{k=1}^{n-r}U_k < \frac{I - \eps/2}{1-\frac{r}{n}}\right) &\leq \Pr\left(\frac{1}{n-r}\sum_{k=1}^{n-r}U_k < I - \eps/4\right) \\
\label{eq:ULLN1}&\leq \Pr\left(\frac{1}{n-r}\sum_{k=1}^{n-r}U_k < \Expt U - \eps/8 \right) \\ 
\label{eq:ULLN2}&= o_{m,t,\xi,\eps}(1), 
\end{align}
where~\eqref{eq:ULLN1} follows from~\eqref{eq:UclosetoI}, and~\eqref{eq:ULLN2} is by virtue of the law of large numbers. 
Furthermore, 
\begin{align}
\nonumber   \Pr\left(\frac{1}{r}\sum_{k=1}^rW_k < -\frac{n\eps}{2r} \right) &\leq  \Pr\left(\frac{1}{r}\sum_{k=1}^r|W_k| > \frac{n\eps}{2m} \right) \\
\label{eq:WMarkov1}  &\leq \frac{2m}{n\eps}\cdot \Expt |W| \\ 
\label{eq:WMarkov2}&= O_{m,t,\xi,\eps}(n^{-1}) , 
\end{align}
where~\eqref{eq:WMarkov1} follows from Markov's inequality, and~\eqref{eq:WMarkov2} is by virtue of Lemma \ref{lem:uw} property \ref{item:wlim}. We therefore obtain that for any $m$ and $\eps$ there are $t$, $\xi$ large enough such that 
\begin{align*}
  \Pr(S_n' < n(I-\eps)) \leq \delta(m) + o_{m,t,\xi,\eps}(1) , 
\end{align*}
where $\delta(m)\to 0$ as $m\to\infty$. Since we can fix $m$ arbitrarily large we have that 
\begin{align*}
  \lim_{n\to\infty}\Pr(S_n' < n(I-\eps)) = 0
\end{align*}
as desired. 
\end{proof}

Finally, combining Lemmas \ref{lem:lower_bound_sn} and \ref{lem:S'achievable} with the definition of $R_n$, we obtain 
\begin{align*}
  \lim_{n\to\infty}\Pr(R_n > I(X;Y)-\eps) &= \lim_{n\to\infty}\Pr(S_n > n(I(X;Y)-\eps)) \\ 
  &\geq  \lim_{n\to\infty}\Pr(S_n' > n(I(X;Y)-\eps)) \\
  &= 1, 
\end{align*}
establishing the theorem.

\appendix

\section{Appendix}
\begin{proof}[Proof of Lemma \ref{lem:hardy_littlewood}]
Define the function $\phi:\mathbb{R}\to\mathbb{R}$
\begin{align*}
\phi(x) = g'(t\mod1)\cdot \ind(x\in[-1,2]) . 
\end{align*}
Let $M\phi(x)$ be the Hardy-Littlewood maximal function \cite[Chapter 7]{rudin1987real} pertaining to $\phi(x)$, i.e., 
\begin{align}
\nonumber  M\phi(x) &\dfn  \sup_{\lambda>0}\frac{1}{\lambda}\int_{x-\lambda/2}^{x+\lambda/2}|\phi(t)|dt \\
& = \sup_{\lambda>0}\Expt|\phi(x+Q_\lambda)| , 
\end{align}
where $Q_\lambda \sim \textrm{Unif}\left(\left[-\tfrac{\lambda}{2},\tfrac{\lambda}{2}\right]\right)$. For any $x\in[0,1)$ we can also write 
\begin{align}
\nonumber  M\phi(x) & \geq \sup_{\lambda\in(0,1)}\Expt|\phi(x+Q_\lambda)| \\ 
\nonumber  & = \sup_{\lambda\in(0,1)}\Expt|g'((x+Q_\lambda) \mod1)| \\ 
\label{eq:Mphi} & = \overline{D}[g(x)] , 
\end{align}
where we have used Lemma \ref{lem:derivative} in~\eqref{eq:Mphi}. Hence
\begin{align}\label{eq:DtoHLMax}
  \Pr(\overline{D}[g(x)] > a) \leq \Pr(M\phi(X) > a). 
\end{align}

The Hardy-Littlewood maximal inequality \cite[Chapter 7]{rudin1987real} implies that for any $a>0$, the following measure-theoretic ``generalized Markov inequality'' holds: 
\begin{align*}
  \left|\left\{x:M\phi(x) > a\right\}\right| &\leq 3a^{-1}\int_{-\infty}^\infty |\phi(x)|dx \\ 
  & = 3a^{-1}\int_{-1}^3 |\phi(x)|dx \\ 
  & = 9a^{-1}\int_{0}^1 |g'(x)|dx . 
\end{align*}
Thus, if $X\sim\textrm{Unif}\left([0,1]\right)$ then 
\begin{align}\label{eq:HLMI}
  \Pr(M\phi(x) > a) \leq 9a^{-1}\Expt \left|g'(X)\right| . 
\end{align}
The proof now follows from~\eqref{eq:DtoHLMax} and~\eqref{eq:HLMI}. 
\end{proof}

\begin{proof}[Proof of Lemma \ref{lem:uw}]
  \begin{enumerate}[label=(\roman*)]
  \item 
    \begin{align*}
      \Pr(L_n\leq u\mid  S_{n-1} = t_0) &= \Pr(L^{(2^{-t_0})}\leq u) \\
      &\leq \sup_{\lambda\in (0,2^{-t}]}\Pr(L^{(\lambda)}\leq u)\\ 
      &= \Pr(\wt{U}\leq u) \\ 
      &\leq \Pr(U\leq u) .
    \end{align*}
  \item Follows similarly. 
  \item Follows similarly to Lemma \ref{lem:Llambda} claim \ref{item:ELpositive}. 
  \item Follows similarly to Lemma \ref{lem:Llambda} claim \ref{item:EL=I}.

\item Write $q(v,y)\dfn \frac{\partial}{\partial v}\left(F^{-1}_{\Theta\mid Y}(v\mid y)\right)$, and note that 
  \begin{align}
\nonumber    \Expt_Y\Expt_V |q(V,Y)| &= \Expt_Y\Expt_V q(V,Y) \\
\nonumber     &= \Expt_Y \left(F^{-1}_{\Theta\mid Y}(1\mid Y)-F^{-1}_{\Theta\mid Y}(0\mid Y)\right) \\
\label{eq:Ee1}    &= 1 . 
  \end{align}
Now, let $w>0$. 
    \begin{align}
\nonumber      \Pr(W \leq -w) &= \sup_{\lambda\in(2^{-t},1)}\Pr\left(L^{(\lambda)} \leq -w\right) \\ 
\nonumber      & \leq  \sup_{\lambda\in(2^{-t},1)}\Pr\left(\inf_{\lambda'\in(0,1)}L^{(\lambda')} \leq -w\right) \\
\nonumber      & = \Pr\left(\log\sup_{\lambda'\in(0,1)}D_{\lambda'}[F^{-1}_{\Theta\mid Y}(V\mid Y)] > w\right) \\
\nonumber      & = \Pr\left(\log\overline{D}[F^{-1}_{\Theta\mid Y}(V\mid Y)] > w\right)\\
\nonumber      & = \Pr\left(\overline{D}[F^{-1}_{\Theta\mid Y}(V\mid Y)] > 2^w\right)\\
\nonumber      &= \Expt_Y\left(\Pr\left(\overline{D}[F^{-1}_{\Theta\mid Y}(V\mid Y)] > 2^w\mid Y\right)\right)\\
\label{eq:HLapp}      & \leq 9\cdot 2^{-w} \cdot  \Expt_Y\Expt_V\left|g(V,Y)\right|\\
\label{eq:Ee1app}      & = 9\cdot 2^{-w}, 
    \end{align}
where in~\eqref{eq:HLapp} we have used Lemma \ref{lem:hardy_littlewood} together with property \ref{P1}, and~\eqref{eq:Ee1app} follows from~\eqref{eq:Ee1}.  Now, 
\begin{align*}
  \Expt |W| &= \Expt\left(\int_0^\infty\ind(W\geq w)dw + \int_0^\infty\ind(W\leq-w)dw\right) \\
&= \int_0^\infty \Pr(W\geq w)dw + \int_0^\infty\Pr(W\leq -w)dw\\
&\leq \int_0^\infty \ind(w\leq \xi)dw + \int_0^\infty 9\cdot 2^{-w}dw\\
& = \xi + 9 \log{e} . 
\end{align*}
Note that the bound is independent of $t$.

   \end{enumerate}
\end{proof}

\begin{proof}[Proof of Lemma \ref{lem:lower_bound_sn}]
Let $A_n\dfn \mathds{1}(S_n < t) = \mathds{1}\left(\sum_{k=0}^{n}L_k < t\right)$. For any $\mu$:
\begin{align}
\nonumber \Pr(S_n<\mu) &\leq \Pr\left(\sum_{k=1}^nL_k < \mu\right) \\
\nonumber &= \Expt_{L^{n-1}}\Pr\left(\sum_{k=1}^nL_k < \mu  \mid L^{n-1}\right) \\ 
\nonumber &= \Expt_{L^{n-1}}\Pr\left(L_n < \mu  - \sum_{k=1}^{n-1}L_k \mid L^{n-1}\right) \\
\label{eq:stochasticLBapp}&\leq \Expt_{L^{n-1}}\Pr\left((1-A_{n-1})U_1 + A_{n-1}W_1 < \mu  - \sum_{k=1}^{n-1}L_k \mid L^{n-1}\right) \\
\nonumber &= \Pr\left((1-A_{n-1})U_1 + A_{n-1}W_1 + \sum_{k=1}^{n-1}L_k < \mu\right) , 
\end{align}
where~\eqref{eq:stochasticLBapp} follows since $(U_1,W_1)$ are independent of $L^{n-1}$, and by virtue of the stochastic lower bound properties \ref{item:u} and \ref{item:u} in Lemma \ref{lem:uw}, according to whether $U_1$ or $W_1$ is selected by  $A_{n-1}$. Iterating the same argument we obtain 
  \begin{align*}
    \Pr(S_n<\mu) &\leq \Pr\left(\sum_{k=1}^n(1-A_{n-k})U_k + \sum_{k=1}^nA_{n-k}W_k < \mu\right) \\
& = \Pr(S_n'<\mu) , 
  \end{align*}
where the last equality follows by noting that $A_k = N_{t,k} - N_{t,k-1}$. This concludes the proof of the Lemma.
\end{proof}

\bibliographystyle{IEEEbib}
\bibliography{C:/Work/latex/ofer_refs_master}

\end{document}